\documentclass[11pt]{article}
\usepackage{mathtools,tikz,parskip,algorithm,amssymb}
\usepackage[noend]{algpseudocode}
\usepackage{geometry}
\usepackage{times}
\usepackage[caption=false]{subfig}
\usepackage{graphicx}
\usepackage{amsmath,amsthm}

\newtheorem{theorem}{Theorem}[section]

\newtheorem{proposition}[theorem]{Proposition}

\usetikzlibrary{positioning,decorations.pathmorphing,shapes,decorations.pathreplacing}

\begin{document}
\title{Comparison-Based Choices}   
\author{  
Jon Kleinberg
\thanks{
Department of Computer Science, Cornell University.
Email: kleinber@cs.cornell.edu
}
 \and 
Sendhil Mullainathan
\thanks{
Department of Economics, Harvard University.
Email: mullain@fas.harvard.edu
}
 \and 
Johan Ugander
\thanks{
Department of Management Science \& Engineering, Stanford University.
Email: jugander@stanford.edu.
}
}
\date{}

\maketitle

\begin{abstract}
A broad range of on-line behaviors are mediated by interfaces in which people make choices among sets of options. A rich and growing line of work in the behavioral sciences indicate that human choices follow not only from the utility of alternatives, but also from the {\em choice set} in  which alternatives are presented. In this work we study comparison-based choice functions, a simple but surprisingly rich class of functions capable of exhibiting so-called {\em choice-set effects}. Motivated by the challenge of predicting complex choices, we study the query complexity of these functions in a variety of settings. We consider settings that allow for active queries or passive observation of a stream of queries, and give analyses both at the granularity of individuals or populations that might exhibit heterogeneous choice behavior. Our main result is that any comparison-based choice function in one dimension can be inferred as efficiently as a basic maximum or minimum choice function across many query contexts, suggesting that choice-set effects need not entail any fundamental algorithmic barriers to inference. We also introduce a class of choice functions we call distance-comparison-based functions, and briefly discuss the analysis of such functions. The framework we outline provides intriguing connections between human choice behavior and a range of questions in the theory of sorting.
\end{abstract}

\let\thefootnote\relax\footnotetext{This work is supported in part by 
a Simons Investigator Award, an ARO MURI grant, a Google Research Grant, Facebook Faculty Research Grants,
and a David Morgenthaler II Faculty Fellowship.}


\newlength{\myfigwidth}
\setlength{\myfigwidth}{0.55\columnwidth}
\newlength{\mynegvspace}
\setlength{\mynegvspace}{0pt}

\newcommand{\sign}{\operatornamewithlimits{sign}}
\def\eps{{\varepsilon}}
\newcommand{\johan}[1]{\textcolor{blue}{#1}}
\newcommand{\johancom}[1]{ \textcolor{blue}{[{\it #1}]}}


\def\Prf{{\rm Pr}}
\def\D{{\mathcal D}}
\def\ev{{\mathcal E}}
\def\evf{{\mathcal F}}

\newcommand{\Prb}[1]{
\Prf\left[{#1}\right]
}

\newcommand{\Prg}[2]{
\Prf\left[{#1}~|~{#2}\right]
}

\newcommand{\Exp}[1]{
E\left[{#1}\right]
}

\newcommand{\Expg}[2]{
E\left[{#1}~|~{#2}\right]
}

\newcommand{\Prbt}[1]{
\Prf\left[\mbox{\em #1}\right]
}

\newcommand{\Prgt}[2]{
\Prf\left[\mbox{\em #1}~|~\mbox{\em #2}\right]
}

\newcommand{\fix}{\marginpar{FIX}}
\newcommand{\todo}{\marginpar{TODO}}
\newcommand{\new}{\marginpar{NEW}}
\newcommand{\Cov}{{\text{Cov}}}
\newcommand{\Var}{{\text{Var}}}
\renewcommand{\Pr}{{\text{Pr}}}
\newcommand{\xhdr}[1]{{\bf #1.}}
\newcommand{\omt}[1]{}


\section{Introduction}

Modern information systems across a broad range of domains
wrestle with increasingly expansive and heterogeneous
corpora of data describing human decisions ---
from online marketplaces such as Amazon and AirBnB
recording diverse purchasing behavior to
search engines and recommendation systems that process 
continuous clickstreams, where clicks reflect choices.

Many of the decisions made in these on-line environments have
a consistent structure: users are presented with a set of options ---
a list of recommended books, search results, flight options,
or users to follow --- and they select one or more items from the set.
The ubiquity of this framework has led to an active line of 
research on the problem of learning ranking functions
\cite{agarwal-ranking-nips}.
There are many formulations in the literature, but the essential
core is to assume that a user has a ranking over alternatives
that is hidden from us, and by observing their choices we would
like to infer this ranking.
This question can be asked at the level of individuals or at
the level of populations containing heterogeneous but related rankings.

\vspace{\mynegvspace}
\xhdr{Choice-set effects}
A central assumption in the work on rankings is that individuals
start from an underlying total order on the possible alternatives.
Roughly speaking, consider an inventory 
$Z = \{z_1, z_2, \ldots, z_n\}$
of possible choices, globally ordered so that a particular user 
prefers $z_i$ to $z_j$ for $i < j$.
When the user is presented with a set of choices $S \subseteq Z$,
they select the item in $S$ that is ranked first according to this global order.
There are many variations on this basic pattern --- there may be
noise in the selection, or options may be unobserved --- but 
the total ordering is a crucial part of the structure.

A long line of work in behavioral science \cite{simonson-choice-context,shafir-reason-based,kamenica-contextual-inference,echenique-money-pump,trueblood2013not}, however,
indicates that people's empirical choices tend to deviate
in significant ways from this picture.
A person's choice will in fact typically depend not just on the total
ordering on the full set of alternatives $Z$, which is often enormous
and implicit, but also on the {\em choice set} $S$ they are
shown as part of the specific decision at hand. Notable examples of such choice-set effects
include the compromise effect, similarity aversion, and the decoy effect (also 
known as asymmetric dominance), and have been observed in diverse contexts
both offline and online;
recent work in machine learning has highlighted the
prevalence of choice-set effects \cite{sheffet-preference-flips,benson2016relevance,yin2014ad,ragain2016pairwise}
and non-transitive comparisons \cite{chen2016modeling,chen2016predicting}, particularly in online search, 
online ads, and online games.

A canonical example of a choice-set effect is the so-called {\em compromise effect}
\cite{simonson-compromise-effect}:~faced 
with the choice between a mediocre option for \$10, 
a good option for \$15, and an excellent option for \$20,
people will tend to choose the good option for \$15;
but faced instead with a choice between the same
good option for \$15, the excellent option for \$20,
and an outstanding option for \$25, 
people will be more likely to choose the excellent option for \$20.
In other words, there is a tendency to favor the option that
compromises between the two extremes in the set that is presented.

More generally, suppose that a set of products exhibit a 
trade-off between two factors --- for example, price vs.~quality, 
power vs.~reliability, or aesthetics vs.~fuel efficiency.
There is a naturally defined one-dimensional set of products
that are on the {\em Pareto frontier}: they are the options
for which no other product is more desirable in both factors.
Consider four products $A$, $B$, $C$, and $D$ that are arranged
in this order along the one-dimensional frontier ($A$ is the 
least powerful but most reliable, and $D$ as the most
powerful but least reliable).
In a wide variety of settings, $B$ will 
be empirically chosen more often when it is presented as one of the three
choices $S = \{A, B, C\}$ than when it is one of the three choices
$S' = \{B, C, D\}$ --- in other words, when it is the compromise option
rather than an extreme option.
By the same reasoning, $C$ will be chosen more often when it is
presented as part of $S'$ than as part of $S$.

One can ask whether this is simply a representational issue, but
in fact it is much deeper. The point is that in any model based
on a total ordering on the alternatives $Z$, a user would prefer
either $B$ or $C$, and this would hold in {\em any} choice set $S$
that contains both of them.
But $B$ and $C$ are options in both the sets $S = \{A, B, C\}$ and
$S' = \{B, C, D\}$, and the aggregate distribution of
choices in favor of $B$ versus $C$ changes across these two sets.
This suggests that the relative ordering of $B$ and $C$ in fact is
not well-defined in isolation, but must be evaluated in the context
of the choice set ($S$ or $S'$).
It matters what else is on offer. 

\vspace{\mynegvspace}
\xhdr{Comparison-based choice functions}
In this paper, we consider a set of basic algorithmic questions
that arise when we seek to infer models of user choice
in the presence of comparison-based choice-set effects.
We focus on a basic formulation in which these effects arise,
and show how our framework makes it possible to derive asymptotic
bounds on query complexity and sample complexity in performing
these types of inference.
We set up the problem as follows.

{\em (1) Embedding of alternatives.}
The full collection of alternatives is a set $U = \{u_1, \ldots, u_n\}$, 
and there is an embedding of $U$ into a space of {\em attributes} $X$
via a function $h : U \rightarrow X$.
For most of our discussion we will take $X$ to be the real line ${\mathbb R}^1$,
so that $h(u_i)$ is a real number; thus most of the time
we can think of the embedded points $\{h(u_1), \ldots, h(u_n)\}$ 
as representing a one-dimensional trade-off continuum as in the
discussion above.
We will assume that the embedding $h$ is not known to us.
Indeed, while we sometimes might know the crucial one-dimensional 
attribute of an item, such as a price, the important
attributes in many settings will not be explicitly presented to us:
a user might have a mental ordering of clothing styles on a spectrum from
``too dull" to ``too ostentatious,''
or restaurants on a spectrum from ``too bland" to
``too exotic," or book or movie recommendations
on a spectrum from ``too much what I'm already reading" to ``too far from my
interests."

{\em (2) Choice functions.}
We focus on an individual who is presented with a subset $S \subseteq U$
of a fixed size $k$ and chooses one element from $S$.
Throughout this discussion we will refer to subsets of $U$ of size $k$
as {\em $k$-sets} of $U$.
We study $k$-sets, rather than subsets of arbitrary/varying size,
both for conceptual clarity and also with the motivation that many
of our motivating applications --- online recommendations and search results
--- often present choices between a fixed number of options.

An individual's selections are represented using a 
{\em choice function} $f$: for each $k$-set $S \subseteq U$, 
we define $f(S) \in S$ as the individual's selection when presented with $S$.
We say that $f$ exhibits {\it choice-set effects} if 
the identity of the set $S$ affects the relative choice between two elements:
specifically, $f$ exhibits choice-set effects if 
there exist $k$-sets $S$ and $T$, and elements
$u_i$ and $u_j$, such that $u_i, u_j \in S \cap T$, and
$f(S) = u_i$ while $f(T) = u_j$.
If such choices can occur then a choice between $u_i$ and $u_j$ 
depends on whether they are presented in the context of $S$ or $T$;
one can view such a contextual effect as a violation of
the independence of irrelevant alternatives (``IIA'') \cite{sen1969necessary}.

We define choice functions here as deterministic for a given individual.
Later in this work we study 
populations composed of a mixture of different choice functions,
and our results there can be interpreted equivalently as applying to the
choices of an individual making probabilistic decisions corresponding to 
a randomization over different choice functions. 
We can contrast such probabilistic choices 
(those expressible as mixtures), 
with random utility models (RUMs)
\cite{baltas-doyle,luce1977choice}:
they lack some of the flexibility of RUMs; but 
discrete choice models such as RUMs on the other hand 
typically entail other restrictions, including
assuming the independence of irrelevant alternatives. 
Identifying the expressive limits of
random mixtures of choice functions
for modeling probabilistic choice is an intriguing open challenge.

{\em (3) Comparison-based functions.}
We focus here on comparison-based choice functions, which 
incorporate the ordinal structure inherent in preference learning,
but which are still rich enough to exhibit choice-set effects.
A choice function $f$ on $k$-sets is {\em comparison-based} 
if the value $f(S)$ can be computed purely using comparisons
on the ordering of the numbers $\{h(u_i) : u_i \in S\}$.
It is not hard to see that for any comparison-based choice function
$f$ on $k$-sets, there is a number $\ell$ between $1$ and $k$ so that 
for all $k$-sets $S$, the value
$f(S)$ is equal to the $\ell^{\rm th}$ ranked element in $S$ according
to the embedding $h(\cdot)$. We can therefore characterize 
any comparison-based choice functions in one dimension as a 
{\em position-selecting choice function} for some position $\ell$ of $k$.

Thus, comparison-based choice functions represent different versions
of the compromise effects discussed earlier: faced with a set $S$
of $k$ options ranked along a one-dimensional spectrum, an
individual would choose the $\ell^{\rm th}$ ranked option.
For example, when $k = 3$ and $\ell = 2$, the individual 
always chooses the middle of three options, much like the
sample instance with choices $A$, $B$, $C$, $D$ discussed earlier.
It is not hard to verify, using examples such as this one, 
that a comparison-based choice function exhibits choice-set 
effects if and only if $2 \leq \ell \leq k-1$.

Choice functions over $k$-sets as we define them here are very general 
mathematical functions, capable of encoding 
an arbitrary choice for every $k$-set. Comparison-based functions in one dimension,
however, are a structured subset of choice functions that provide a useful
abstraction of many decision contexts. We are not aware of any prior
inference work on general comparison-based functions, which is
surprising given the central role of binary 
comparison in the fundamental problem
of sorting. We focus on comparisons in one dimension.
Comparison functions in higher dimensions can have a significantly
more complex structure;
 extending our results to higher dimensions
is beyond the scope of the present work.

A short-coming of comparison-based functions is that they 
lack the ability to grasp ``similarity,'' an important aspect of 
choice-set effects such as similarity aversion and the decoy effect, 
and one that requires a notion of distance beyond ordinal comparison.
As a step towards extending our framework to model such effects,
we also study {\em distance-comparison-based choice functions} 
that model choices according to distance comparisons, 
possibly in high-dimensional latent spaces.

\vspace{\mynegvspace}
\xhdr{The present work: Asymptotic complexity of inference}
There are many questions that one can consider for models
of comparison-based choice functions; here we study a basic
family of problems that are inherent in any inference procedure,
and which form an interesting connection to fundamental questions in sorting.

The basic problem we study has the following structure:
we observe a sequence of choices of the form $(S,f(S))$,
and at the end of this sequence we must 
correctly report the value of $f(S)$ for all (or almost all)
$k$-sets $S$.  
The question is how few observations $(S,f(S))$ we need in order 
to achieve various measures of success.
We ask this question for different models of how the observations
are generated.  
We first consider
active queries, in which we can choose $S$ and receive
the value of $f(S)$; we investigate the potential for efficient inference both
when a single individual is making choices with a fixed function $f$,
and for population mixtures of different choice functions.
As a second model we consider passive queries, a model 
whereby a stream of pairs $(S,f(S))$ is generated uniformly at
random over possible $k$-sets $S$, without the control over $S$ offered
by the active query model.

For active queries, a natural baseline for understanding the problem 
formulation is the problem of sorting, which precisely consists of the case 
$k = 2$ and $\ell = 1$.
A comparison-based sorting algorithm asks about a sequence of pairs $S = \{u_i, u_j\}$,
and for each such pair it is told the identity of the preferred element.
An efficient sorting algorithm given $O(n \log n)$ such queries can
learn the sorted order of the elements, and thus can answer 
queries of the form $f(S)$ for arbitrary pairs after learning the sorted order.

We show that all comparison-based choice functions in one dimension 
exhibit the
same $O(n \log n)$ active query complexity: for every $\ell$ and $k$, 
there is an algorithm that can ask about a sequence of
$O(n \log n)$ $k$-sets $S$
as queries, and from the values of $f(S)$ it is then prepared to
correctly report $f(S)$ for all $k$-sets $S$.
Roughly speaking, the algorithm works by first identifying a small
set of elements that cannot be the answer to any query,
and then it uses these elements as ``anchors'' on which to
build sets that can simulate comparisons of pairs.
Note that as the set size $k$ grows, the $O(n \log n)$ queries form a 
smaller and smaller fraction of the set of all $n \choose k$ 
possible $k$-sets.

We then consider active queries for population mixtures, in which
different people use different comparison-based choice functions on
$k$-sets, and when we pose a query $S$, we get back the answer
$f(S)$ for an individual selected uniformly at random from
the population.  We show that for a fixed but unknown vector of 
probabilities for each choice function
under some natural
non-degeneracy conditions we can determine $f(S)$ for each
segment of the population after $O(n \log n)$ active queries.
The algorithm here uses random sampling ideas combined with
techniques for {\em sorting using noisy comparators}
\cite{feige1994computing}.

For passive queries, where an algorithm is
presented with a stream of values $(f(S),S)$ for randomly
selected sets $S$,
the question is how long we need to observe values from the
stream before being able to compute $f(S)$ for all (or almost all) 
$k$-sets $S$.
We show that for comparison-based choice functions $f$ that
exhibit choice-set effects (with $\ell$ between $2$ and $k-1$),
we can do this after observing $o(n^k)$ values from the stream 
with high probability; 
for an arbitrary $\eps > 0$ and $\delta > 0$, 
we can, with probability at least $1 - \delta$,
determine $f(S)$ for
at least a $1 - \eps$ fraction of all $k$-sets $S$.
Our analysis here builds on a sequence of combinatorial results
on {\em sorting in one round}, culminating in an asymptotically
tight analysis of that problem by Alon and Azar \cite{bollobas1981sorting,alon1988sorting}.

Finally, we consider how our results for comparison-based
functions apply to distance-comparison-based choice functions 
in which the geometry of the alternatives'
embedding plays a consequential role.
A range of earlier work have made use of the ambient space in 
different ways, including methods such as {\em conjoint analysis}
\cite{baltas-doyle,green-conjoint-survey}.
Here, we consider the effect of performing comparisons among the
pairwise distances between alternatives.
This consideration enables us to reason about elements that are either
{\em central} or {\em outliers}
in a comparison-based fashion, providing a plausible model
for the choice-set effect commonly known as {\em similarity aversion}.
This line of inquiry into metric embeddings also connects our results with recent 
research in the learning and crowdsourcing literature on learning
stochastic triplet embeddings 
\cite{van2012stochastic} and inferences using 
the crowd median algorithm
\cite{heikinheimo2013crowd}.

\section{Active query complexity}

We will first focus on {\em active} query algorithms
that may choose queries sequentially based 
on the results of previous queries, 
and our goal here is to develop algorithms that
after performing a small number of queries
can take an arbitrary $k$-set $S \subseteq U$ 
from a universe of $n$ alternatives and correctly output $f(S)$.

We have previously noted that every
comparison-based choice function over
a one-dimensional ordering takes 
the form of a position-selecting choice function 
choosing the $\ell^{\rm th}$ alternative of $k$
for some ordering; we will denote this function by $q_{k,\ell}$.
Our results in this section show that for any fixed $k\ge 2$ 
and any fixed but unknown $\ell \in \{ 1, ..., k\}$, we can in fact
learn the output for any input using an efficient two-phase algorithm 
that performs only $O(n \log n)$ queries.
This algorithm determines $\ell$ (up to a reflection we discuss below) 
by the time it terminates,
but we also give a simple algorithm that recovers $\ell$ (up to the same reflection)
more directly in just $O(1)$ queries (without learning the output for all inputs).
Thus, we establish that learning the position $\ell$ of a comparison-based
choice function does not require running a comprehensive recovery algorithm. 

For recovering the choice function, we note that the
orientation of the embedding is inconsequential.
This symmetry is clear when considering the simplest example of $k=2$, where
$f$ is either a max selector $q_{2,1}$ or a min selector $q_{2,2}$.
In order to deduce $f(S)$ for every $S$ we needn't 
know whether $f$ is choosing the max or min, as we will simply
learn an embedding aligned with our selector.
More generally, we have no way or need to distinguish between
$q_{k,\ell}$ over a given embedding and 
$q_{k,k-\ell+1}$ over the same embedding reversed,
as they result in the exact same choices.

\vspace{\mynegvspace}
\xhdr{Recovering choice functions}
The algorithm we propose consists of two phases. In the first phase,
we use $O(n)$ queries to identify a set of ``ineligible'' alternatives. 
In the second phase, we use those ineligible alternatives to 
``anchor'' the choice set down to a binary comparison
between two eligible alternatives, allowing us to 
determine the ordering of the eligible alternatives 
in $O(n \log n)$ queries using comparison-based sorting.

\begin{theorem}
Given a choice function $f$ over $k$-sets $S$ in a 
universe of $n$ alternatives, 
there is an algorithm that 
recovers $f$ after $O(n \log n)$ queries,
meaning that after this set of queries it can output $f(S)$ for any $S$. 
\end{theorem}

We prove this theorem by describing the algorithm together with
its analysis, divided into its two phases of operation.
For the first phase, we identify the ineligible elements.
In the second phase, we simulate pairwise comparisons.

\begin{figure}[t]
\centering
\includegraphics[width =\myfigwidth]{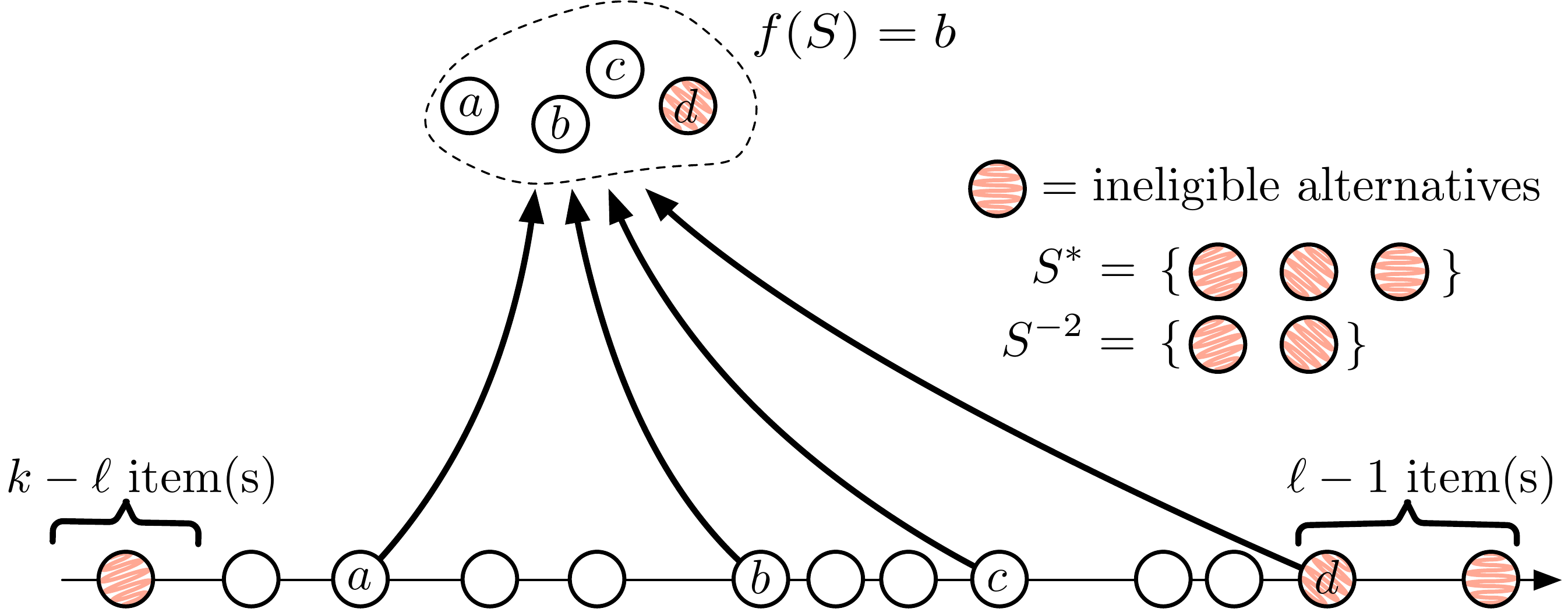}
\vspace{-3pt}
\caption{
For a given comparison-based choice function $f$ over an unknown embedding,
there are ineligible alternatives that will never be chosen,
illustrated here for $k=4$ and $\ell=3$. The elements are embedded in a $1$-dimensional 
space, where the $\ell-1$ maximal and $k-\ell$ minimal elements can 
never be chosen. 
}
\vspace{0pt}
\label{f:fig1}
\end{figure}

\begin{proof}
We begin the first phase by observing that for a position selector $q_{k,\ell}$ making 
comparison-based choices over $k$-sets $S \subseteq U$ with $|U|=n$,
there are $k-1$ alternatives that will never be chosen.
For max-selectors $q_{k,k}$ these are the $k-1$ minimal alternatives,
for min-selectors $q_{k,1}$ these are the $k-1$ maximal alternatives, and for 
general $q_{k,\ell}$ these are the $\ell-1$ maximal and $k-\ell$ minimal alternatives
from the embedded order.

To find these elements, we run a simple {\em discard algorithm}.
Commencing with $k$ arbitrary initial alternatives in a choice set $S_1$, 
we query the choice function $f$ and learn the choice $f(S_1)$.
We then construct our next choice set by discarding the previous choice
and selecting a new arbitrary alternative $u \in U \setminus S_1$ to form
$S_2 = (S_1 \setminus  f(S_1)) \cup u $. We query for $f(S_2)$, and
repeat this discard procedure for $n-k+1$ queries. 
After the last query we will have exhausted $U$ and learned that 
$S^* = S_{n-k+1} \setminus f(S_{n-k+1})$ are precisely the $k-1$ alternatives
in $U$ that can never be selected by $f$.
Having found the $k-1$ ineligible alternatives $S^*$, we 
arbitrarily select $k-2$ alternatives from $S^*$ to form a set $S^{-2}$.
See Figure 1 for an illustration of this procedure. 

In the second phase of the algorithm we use the set $S^{-2}$
constructed above as padding for a choice set
construction with two open positions: notice that for any two
alternatives $u_i$ and $u_j$, 
a query for $f(\{u_i, u_j\} \cup S^{-2})$ will amount to binary comparison
between $u_i$ and $u_j$.
If $\ell=1$ or $\ell=k$, we know that this will be a max or a min selection,
respectively. However, for intermediate values of $\ell$ we do not
yet know whether such binary comparisons will be a max or a min selection, 
as it depends on what $k-2$ elements from $S^*$ were chosen
to construct $S^{-2}$. 
There exists no basis for choosing $S^{-2}$ from $S^*$ in a
non-arbitrary way, as we have no way of
ever choosing any of the alternatives in $S^{-2}$.
However,  as noted earlier in the case of 
recovering choice functions for $k=2$,
we can simply assume that the comparator we have constructed
is selecting the max of $u_i$ and $u_j$, 
and then learn the order of the eligible alternatives
 from binary comparisons as oriented by that comparator. 
After $O(n \log n)$ padded comparison queries, 
we obtain a sorted order of the eligible alternatives.

To conclude this second phase, 
we perform a single additional query to determine $\ell$, the position
being selected,
which is as yet unknown and unused. 
By querying for what we've come to 
suppose are the first $k$ eligible elements, 
the chosen alternative will identify the $\ell^{\rm th}$ position
in the recovered embedding of the eligible alternatives.
We again emphasize that we 
have no way of distinguishing between
$q_{k,\ell}$ over a given embedding and 
$q_{k,k-\ell+1}$ over the same embedding reversed,
but the direction of the ordering is nonessential, as
the choices are identical for all choice sets. 
Having learned the embedded order of the eligible
alternatives in $U$ as well as the 
position $\ell$ being selected,
we can now discern the choice $f(S)$ for any $k$-set $S \subseteq U$.
\end{proof}

As a further point, it is notable that no part of this algorithm 
depends on $\ell$, as it is only learned at the end of the final query,
meaning that the entire algorithm is in fact 
indifferent to whether an individual is selecting 
minima, maxima, or intermediate positions.

Lastly, we observe that the classic lower bound on sorting applies 
as a lower bound here as well, for fixed $\ell$ and $k$: 
there are $(n-k+1)!/2$ possible permutations of how the eligible alternatives
can be embedded (ignoring the $k-1$ ineligible alternatives). 
Each query cements at most $k-1$ 
relative orderings, at most reducing the number of feasible orderings by 
a factor of $k-1$. We thus require at least $\log_k((n-k)!/2)$
queries to learn a choice function, yielding 
a lower bound of $\Omega(n \log n)$, as for sorting. 

\vspace{\mynegvspace}
\xhdr{Type classification}
It can be valuable to learn how a choice function 
makes comparisons without necessarily learning the ordering 
implicated in the comparisons.
Here we establish that we can learn $\ell$, up to the noted reflection equivalence, 
using $O(k)$ queries without learning the order of the alternatives. 

\begin{proposition}
By querying for all $k$-set subsets of an arbitrary $(k+1)$-set,
we can determine $\ell$, up to reflection.
\end{proposition}

\begin{proof}
The algorithm performs $k+1$ queries for all the subsets of size $k$, 
where each of these subsets can be defined by what element is excluded. 
The $k+1$ alternatives have an unknown internal ordering, 
and we let $u_i$ denote the $i^{\rm th}$ alternative in order.
From these queries, there can be only two possible elements selected 
as the output across the $k$ different subsets: 
either the chosen alternative in the $(\ell+1)^{\rm th}$ ordered position 
(when the excluded element is one of $u_1,\ldots,u_\ell$) 
or the chosen alternative in the  $\ell^{\rm th}$ ordered 
(when the excluded element is one of $u_{\ell+1},\ldots,u_{k+1}$).
Of these $k+1$ $k$-sets, the $(\ell+1)^{\rm th}$ element will be chosen $\ell$ times, 
and the $\ell^{\rm th}$ element will be chosen $k-\ell+1$ times. 

Thus, for any $k$, in order to determine the position $\ell$ of
a position-selecting choice function, one can
simply take an arbitrary set of $k+1$ elements 
and query for all subsets of $k$. 
Then $\ell$ (up to reflection) is the frequency of the less frequent 
of the two response elements.
\end{proof}


\section{Population mixtures}

The results in the previous section assumed that all queries were evaluated by the
same deterministic choice function $f$, where $f(S)$ was consistently 
selecting the $\ell^{\rm th}$ ordered alternative from within each $k$-set $S$.
In this section, we show that we can recover choice functions 
in a more general setting, where the position $\ell$ chosen
by $f$ is drawn from a distribution over possible positions. 

This setting covers two generalizations of our active query results.
First, it describes situations where a single individual is being
 repeatedly queried, and may exhibit compromise effects 
for some queries, independently at random. Second,
it describes situations where the queries are handled by
a population of individuals that all base their choices on the
same ordered embedding, but differ in what position
within an ordering that their choice functions select for. 
We show that we can still recover choice functions in this
setting using $O(n \log n)$ queries, almost surely.

Recalling one of our motivating examples from the introduction, this
is a plausible reality: faced with a choice of restaurants
ranging from ``too bland'' to ``too exotic,'' some individuals 
will compromise, while some will optimize for one of the
extremes. In online settings the choices being made often
come from heterogeneous populations, and we wish 
to develop an algorithm that can still recover choices in these
settings.

We define a {\em mixed choice function}
as a comparison-based choice function on $k$-sets $S \subset U$
where the $\ell^{\rm th}$
ordered element in $S$ is selected independently at random
with probability $\pi_\ell$. 
A mixed choice function is then completely 
defined by an ordering of $U$ and a probability 
distribution $(\pi_1,...,\pi_k)$ over position-selecting choice functions
with $\sum_{i=1}^k \pi_i = 1$.
To avoid degeneracies, we require
that $\pi_\ell > 0$, $\forall \ell$,
and also require constant separation
between the probabilities,
$|\pi_\ell- \pi_{\ell'}| > \gamma$, $\forall \ell,\ell'$ with $\ell \ne \ell'$, 
for some constant $\gamma>0$. 

In its most basic instance, a mixed choice function 
over $2$-sets is simply a noisy binary comparator over an ordering,
with $(\pi_1,\pi_2)=(p,1-p)$ for some probability $p$. 
In this case, 
results from the sorting literature 
contribute that as long as $p$ is bounded away from
$1/2$, the order can be recovered in $O(n \log n)$ queries
almost surely \cite{feige1994computing}. 
Our contribution in this section is to 
generalize this result to arbitrary mixtures of comparison-based choice functions.

We begin by showing that for any $\epsilon >0$, 
we can recover the mixture probabilities $(\pi_1,...,\pi_k)$
of a mixed choice function with probability at least
$1-\epsilon$ in a number of queries that is $O(1)$ 
in $n$, the size of $U$. We then
show that we can indeed recover any 
mixed choice function
$f$ using $O(n \log n)$ queries with probability
at least $1-\epsilon$.

\vspace{\mynegvspace}
\xhdr{Recovering mixture probabilities} 
We begin by showing that we can recover the mixture probabilities using a 
modified version of the algorithm we presented for recovering $\ell$
in an active query framework.

\begin{theorem}
\label{mixtureprop1}
Let $f$ be a mixed comparison-based choice function over $k$-sets in a universe of $n$ alternatives.  Let $\pi = ( \pi_1,...,\pi_k)$ be the mixture distribution of $f$, and let $\pi_\ell > 0$, $\forall \ell$, and $|\pi_\ell- \pi_{\ell'}| > \gamma$, $\forall \ell,\ell'$ with $\ell' \ne \ell$, for some constant $\gamma>0$. 

For any $\epsilon>0$ and $\delta \in (0,\gamma/2)$ there exists a constant $C>0$ such that an algorithm using $C$ 
queries can recover $\pi$, meaning that it outputs a probability vector $p$ for which 
$$Pr(max_\ell |\hat \pi_\ell - \pi_\ell | \le \delta) \ge 1 - \epsilon$$ 
holds for either $\hat \pi = (p_1,...,p_k)$ or $\hat \pi = (p_k,...,p_1)$.
\end{theorem}

\begin{proof}
Our strategy is to study a single $(k+1)$-set $S^+$ closely, and by querying each $k$-set within $S^+$
sufficiently many times we can recover the mixture distribution with the required precision. 
We index the $k+1$ $k$-sets within $S^+$ as  $S_1,...,S_{k+1}$.

We define the random variables 
\begin{align}
X_{i,j}^u = 
\begin{cases}
1 \text{ if query } i \text{ of subset } S_j \text{ returns } u \in S_j  \\
0 \text{ otherwise,}
\end{cases}
\end{align}
where $\mathbb E [ X_{i,j}^u ] = \pi^j_u$ for a 
vector of probabilities $\pi^j = (\pi^j_1,...,\pi^j_k)$ 
that is an unknown permutation of $(\pi_1,...,\pi_k)$.

We first show that for each subset 
we can recover the complete set of probabilities
with the specified precision and correctly
ordered by their relative frequency. 
In a second stage we will 
then align the permuted $\pi^j$
 vectors and thereby recover 
$\pi$ itself (or $\pi$'s reflection).

For each $S_j$, we consider the estimates $\hat \pi^j_u = \frac{1}{C} \sum_{i=1}^{C} X^u_{i,j}$, and we show that this is correct to within the requested error tolerance if we choose $C$ large enough.
Recalling that $\delta < \gamma / 2$,
a two-sided Chernoff bound tells us that:
\begin{align*}
\Pr \left ( \left | \frac{1}{C} \sum_{i=1}^{C} X_{i,j}^u - \pi^j_u \right | \ge \delta^j_u \pi^j_u \right ) 
\le 2\exp \left (
- \frac{ (\delta^j_u)^2 }{2+\delta^j_u} C \pi^j_u
\right ),
\end{align*}
for which we can set $\delta^j_u=\delta/\pi^j_u$ and use $\pi^j_u \le 1$, $\forall j,u$ 
to obtain
\begin{align*}
\Pr \left ( | \hat \pi^j_u - \pi^j_u  | \ge \delta \right ) 
\le 2\exp \left ( -
\frac{ \delta^2}{2 + \delta } C
\right ).
\end{align*}
Let us define the bad events $\mathcal E_{j,u} = \{ | \hat \pi^j_u - \pi^j_u  | \ge \delta \}$,
when we failed to recover $\pi_u^j$ (the probability of selecting $u$ from subset $S_j$)
within the specified precision. 
By taking the Union Bound across the sets $S_j$ and their alternatives $u$, 
we can set $C>\frac{2+\delta}{\delta^2} \log(\frac{\epsilon}{2k(k+1)})$ to obtain
\begin{align*}
\Pr  \left ( \max_{u,j} \left | \hat \pi^j_u - \pi^j_u \right | \le \delta  \right ) & \ge
1 - \sum_{j=1}^{k+1}  \sum_{u=1}^k \Pr  \left ( \left | \hat \pi^j_u - \pi^j_u \right | \ge \delta \right ) 
> 1- \epsilon.
\end{align*}
The remaining challenge is
to determine what probability corresponds to what choice position,
and we will show
that recovering this correspondence is guaranteed if the bad events did not occur.

 We recover the correspondence using an alignment procedure: 
 since $\delta < \gamma/2$,
inside each $S_j$ we can order the elements from most frequently
selected to least frequently selected. In this order the $i^{\text{th}}$
most frequent choice inside each vector $\hat \pi^j$ 
will correspond to the same underlying position selection for each $S_j$. 
The most frequently selected element from each $S_j$
will have been selected by the $\ell$-selector with the largest $\pi_\ell$ in $\pi$,
and so on, for each selection frequency.

To discern what frequency corresponds to what $\ell$,
we look across the subsets $S_1,\ldots,S_{k+1}$ and identify the 
two probabilities $\pi_1$ and $\pi_k$ by the fact that they will
be the only two frequencies that have
selected some element $u$ only once at that frequency, 
and another element $v$ the remaining $k$ times at that frequency.
In other words, 
the alternatives selected at those frequencies will have the structure $\{u, v, ...., v\}$
for some $u$ and $v$, and these must be the frequencies $\pi_1$ and $\pi_k$. 
See Figure S1 for an illustration in the case of $k=3$. 
We arbitrarily assign the more frequent of these two
frequencies to be $\pi_1$, recalling that we're only 
looking to recover $\pi$ up to reflection. The other frequency
is assigned to be $\pi_k$.

Lastly, for each of these two end frequencies $\pi_1$ and $\pi_k$
we take the ``other'' element $v$ (that made up $k-1$ of the $k$ choices at 
that frequency), and find what frequency corresponds to 
when each $v$ was selected exactly twice, 
identifying the probabilities $\pi_2$ and $\pi_{k-1}$. 
Since $\pi_\ell>0$ for all $\ell$, we can
continue this overlap procedure for all the frequencies,
allowing us to identify the probability $\pi_\ell$ corresponding
to each position $\ell$. 
This alignment procedure fails
only if one of the earlier bad events failed, and so we can return any one
of the $\hat \pi^j$, which all had the necessary precision, 
and reorder it by our alignment procedure
to produce our output $p$.
\end{proof}

\begin{figure}[t]
\centering
\includegraphics[width =0.8\myfigwidth]{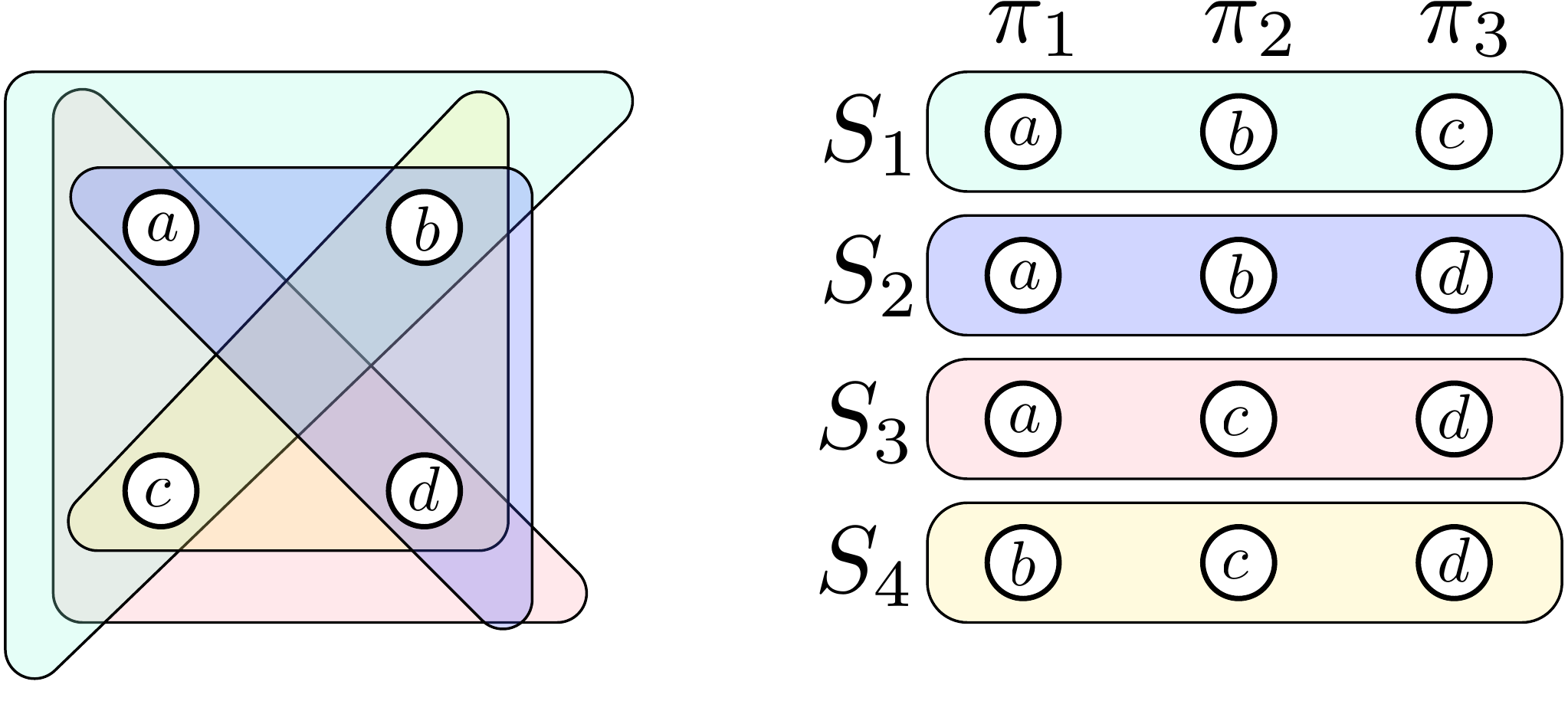}
\vspace{-3pt}
\caption{
A graphical sketch of the population alignment algorithm in Theorem 2, 
shown for $k=3$. 
Left: a $4$-set contains four $3$-sets, $S_1, S_2, S_3, S_4$. Right: 
by querying for each $3$-set sufficiently often, we can
identify the frequencies $\pi_1, \pi_2, \pi_3$ with which each element occurs
and then identify
$\pi_1$ and $\pi_3$ by the fact that their frequencies corresponds to the 
same element in all but one of the $3$-sets.
}
\label{f:fig2}
\end{figure}

\vspace{\mynegvspace}
\xhdr{A recovery algorithm for mixtures}
We now show that in this setting where queries
are handled by a mixed choice function---either representing
a random individual or a random mixture of individuals---we
can recover the choice function in $O(n \log n)$ queries
with high probability, matching the asymptotic query complexity 
of the non-population case for any fixed error probability.
We run our recovery algorithm
against the mixture population, but for the purposes
of recovering $f(S)$ for every $S$ 
we assume that $\ell$ is then known. 

\begin{theorem}
Let $f$ be a mixed comparison-based choice function over $k$-sets in a universe of $n$ alternatives.  Let $\pi = ( \pi_1,...,\pi_k)$ be the unknown mixture distribution of $f$, and let $\pi_\ell > 0$, $\forall \ell$, and $|\pi_\ell- \pi_{\ell'}| > \gamma$, $\forall \ell,\ell'$ with $\ell \ne \ell'$, for some constant $\gamma>0$. 

Let $f_1,...,f_k$ be pure versions of $f$, where $\pi_\ell=1$ for $f_\ell$, 
over the same ordered embedding.
For any $\epsilon>0$
there exists
an algorithm using $O(n \log n)$ queries to $f$ that
with probability at least $1-\epsilon$
can recover 
$f_\ell(S)$ for all $\ell$, all $S$.
\end{theorem}

Our proof is constructive. The strategy for this 
algorithm is to first run our procedure from Theorem 2
for recovering $\pi$ using a constant number of queries.
We then focus on max-selectors, and run a modified 
discard algorithm (from Theorem 1) to identify the alternatives 
that are ineligible to the subpopulation of max-selectors. 
We use these ineligible alternatives for the max-selectors to create
an anchored $k$-set that can furnish a noisy binary comparator. 
By employing results for sorting under 
noisy comparisons \cite{feige1994computing}, we obtain
the order over the alternatives that are eligible to a max-selector.
If we only needed to recover $f_k(S)$ for every $S$, where $f_k$ is
the max-selector, we would be done. In order to also
learn the order of the max-ineligible alternatives, we run a second
sort using a query anchored with min-ineligible alternatives, thereby
also recovering the order of the max-ineligible
alternatives.

\begin{proof}
We begin by running 
the algorithm from Theorem 2 with a sufficiently large 
number of queries $C_1$ to obtain an
estimate $\hat \pi$ of $\pi$, such that
$
\Pr  ( \max_{\ell} \left | \hat \pi_\ell - \pi_\ell \right | \le \gamma/2  ) 
 > 
1- \epsilon/5
$ where $C_1$ is a constant in $n$ for every $\epsilon$ and $\gamma$.
The usual caveats for reflection of the mixture vector apply, and we learn $\pi$
for one of the orientations, arbitrarily chosen. 
We focus on our estimate of $\pi_k$, the probability that a query
is answered by a max-selector. 

We now initiate our standard discard algorithm,
which we will use to find max-ineligible alternatives, starting 
with an arbitrary $k$-set $S \subset U$.
In each discard step, we plan to identify the max element by posing enough queries to the population
that the frequency with which we observe the maximum element is close enough to $\pi_k$.
Let $\mathcal E_i$ be the bad event that discard round $i$ incorrectly identified
the max element in the set ($\pi_k$ was not recovered with sufficient precision), 
and let $\mathcal D$ be the 
success event whereby no bad events occur and 
the discard algorithm terminates with the alternatives that are ineligible 
for the max-selectors.  
Using $C_2$ queries in each round $i$ of the discard algorithm 
we can bound $\Pr[\mathcal E_i]$ using two-sided Chernoff bounds, 
followed by the Union Bound across the
$n-2$ rounds of our discard algorithm to obtain:
\begin{align*}
\Pr[ \mathcal D ] 
\ge 1 - \sum_{i=1}^{n-2} \Pr [ \mathcal E_i ] 
\ge 1 - 2(n-2) \exp \left (-\frac{\gamma^2}{8+2\gamma} C_2 \right ).
\end{align*}
Selecting $C_2 > \frac{8+2\gamma}{\gamma^2} \log (\frac{10}{\epsilon})\log(n) = C_2' \log(n)$, 
we obtain that $\Pr[ \mathcal D ] > 1- \epsilon/5$ 
with $C_2'n \log n$ queries across all the rounds of this discard algorithm, where $C_2'$ is 
constant in $n$.

Conditional on succeeding thus far, we can now create a padded
$k$-set from the max-ineligible alternatives $S^*_{\text{max}}$ that were never selected during our 
discard algorithm, 
forming a binary comparator for any $u$ and $v$ from a $k$-set of the form $\{u,v\} \cup S^{-2}_{\text{max}}$
for any $(k-2)$-set $S^{-2}_{\text{max}} \subset S^*_{\text{max}}$.
The max alternative will be selected with probability $\pi_k$, and
the second most maximal alternative will be selected with probability $\pi_{k-1}$.
Our goal is to hand this comparator off to the 
Feige et al.~comparison-based query 
algorithm \cite{feige1994computing}, but we note that the 
Feige et al.~algorithm requires a binary comparator that 
returns each element
with probabilities $\{p,1-p\}$ for some $p>1/2$. 
Our comparator, meanwhile,
will ``fail'' with constant probability $\Delta=1-\pi_k-\pi_{k-1}$, 
in the sense that 
some alternative $u^* \in S^{-2}_{\text{max}}$ will be selected that is
not one of the elements $u$ or $v$ being subjected to the 
noisy binary comparison.
If one
of $\pi_k$ or $\pi_{k-1}$ exceed $1/2$ then we can apply
the Feige et al.~procedure as-is, as we
can simply bundle any choice $u^* \in S^{-2}_{\text{max}}$ as a
choice of the less frequent of the 
two positions, producing the necessary binary comparator. 
If the selection probabilities $\pi_k$ and $\pi_{k-1}$ 
are both less than $1/2$, we 
can modify our comparator as follows.

For each step $s$ that the Feige et al.~algorithm
performs a query,
we repeat the query until we first see one of $k$ or $k-1$. 
This means that the outcome space is not 
$ \{ k, k-1,\text{fail} \} $,
but just $ \{k,k-1\} $, 
one of which will have a fixed probability $p > 1/2$.
Let $C_s$ be a random variable that takes a value 
equal to the number of queries
needed before first seeing $k$ or $k-1$. 
Seeing as $C_s$ is geometrically distributed, 
we have $\mathbb E[C_s] = 1/\Delta$, the inverse of the
failure probability.
But now the total number of queries being run throughout the 
Feige et al.~algorithm is $\sum_{s = 1}^{D n \log n} C_s$, 
for a constant $D$
large enough to control the Feige et al.~error 
to at most $\epsilon/5$, 
and so in expectation it takes
$
\mathbb E[\sum_{s = 1}^{D n \log n} C_s] 
= (1/\Delta) D n \log n 
$ queries.
By Markov's inequality the probability that the number of queries 
exceeds its expectation by a multiplicative factor of $5/\epsilon$
is $\Pr ( \sum_{s = 1}^{D n \log n} C_s \ge \frac{5D}{\epsilon \Delta} n \log n ) \le \epsilon/5.$

This modification of Feige et al.~will 
then still only take $C_3 n \log n$ queries, for $C_3= 5D/(\epsilon \Delta)$,
to succeed in sorting the max-eligible alternatives, with two possible bad events:
the failure of the Feige et al.~algorithm, with probability at most $\epsilon/5$,
and the number of queries exceeding the query budget, 
with probability at most $\epsilon/5$.
This latter event is only a concern when both $\pi_k$ and $\pi_{k-1}$ are less than $1/2$.

If we only needed to recover the choices of a max-selector, we would be done.
In order to recover all queries for any positions $\ell$, however, we also need 
to know the order of the 
max-ineligible alternatives. To recover such choices, we switch our focus from
the max-selector subpopulation to the min-selector subpopulation.
As a final step
we repeat the above procedure using an anchoring set $S^*_{\text{min}}$ 
that contains the $k-2$ maximum alternatives of the 
ordering we just recovered. We then run the 
Feige et al.~procedure to
order the alternatives in $U_{\text{scrap}} = u_{n-k-2} \cup S^*_{\text{max}}$.
Here  $u_{n-k-2}$ is the last element of the order recovered above, 
included to orient the sort order with regard to the overall sorted order,
and $S^*_{\text{max}}$ was the set of max-ineligible alternatives that we are
trying to order.
Since $|U_{\text{scrap}}|=k-1$, we can 
drive this error probability below $1-\epsilon/5$ 
with a large constant $C_4$ that does not depend on $n$, thereby also recovering the order 
over the max-ineligible alternatives.

We have outlined five separate bad events that would make the algorithm fail.
By taking the Union Bound over the five bad event probabilities, each controlled to have an
error $\le \epsilon/5$, the overall algorithm performs four sets of queries (where there are two bad events associated with the third set of queries) and succeeds with probability at least $1-\epsilon$
after $C_1 + C_2' n \log n + C_3 n \log n + C_4$ queries,
asymptotically $O(n\log n)$ as specified.
\end{proof}


\section{Passive query complexity}

We now shift our attention to {\em passive} query algorithms, 
algorithms that pose all their queries at once
without the sequential benefit of previous responses. 
This framework is motivated by a common 
scenario in the study of large datasets of recorded choices
from online marketplaces, 
ranking algorithms, and recommendation systems.
These systems are teeming with choices being made 
over sets of alternatives,
but often one cannot adaptively guide the queries
towards our goal of recovering choice functions. What if
we simply had a large corpus of choices by an 
individual over random subsets?
Could we still recover
the choice function? If so, how few queries would we need?

In this section we analyze a model for this process, in which 
we are not able to select queries but instead have to watch them passively.
We focus on the output of a single $f$; handling population mixtures
in a passive query model is an interesting open question.

\vspace{\mynegvspace}
\xhdr{Model of passive choice streams}
We analyze the following process for generating and analyzing
a stream of random queries that arrive over time.
For each possible $k$-set $S$, we define a Poisson process of rate $\alpha$:
we draw a length of time $t$ from an exponentially distributed
random variable of rate $\alpha$ (with density $\alpha e^{-\alpha x}$),
and after this time $t$ elapses, we put the set $S$ into the stream.
We iterate this, repeatedly drawing a random length of time $t$
and putting $S$ into the stream when $t$ is reached.

We run the process for each $k$-set $S$ simultaneously, 
creating a merged stream of $k$-sets as they are selected
and put into the stream.
Now, suppose we observe the merged stream over the interval $[0,T]$.
For any fixed $k$-set $S$, the probability that it appears in
the interval is $1 - e^{-\alpha T}$, a quantity that we'll
refer to as $p_T$.
An algorithm is provided with the value $f(S)$ for each $k$-set
that appears in the stream during $[0,T]$. After seeing these
values it must correctly report the
value of $f(S)$ for almost every $k$-set.
We want an algorithm that can achieve this goal while
keeping $T$ as small as possible.

Our streaming model is arguably the simplest plausible
model for generating collections of (choice set, choice)
decisions that are independently drawn with replacement.
Other seemingly simple models 
that query for a random set of $k$-sets without replacement
exhibit an implausible slight dependence between decisions in the
decision collection.

\vspace{\mynegvspace}
\xhdr{A recovery algorithm for passive streams}
Our main result for this streaming model of passive queries
is that for $2 \leq \ell \leq k-1$
there is an algorithm able to
provide the desired performance guarantee for a choice of $T$
with $p_T \rightarrow 0$, which is to say 
that the probability of seeing any 
fixed $k$-set $S$ even once goes to zero. Equivalently, 
this result implies that we can recover $f(S)$ 
for almost every $k$-set after only seeing
a tiny fraction of all possible $k$-sets.

\begin{theorem}
Let $f$ be a choice function over $k$-sets in a 
universe of $n$ alternatives
 that
selects the $\ell^{\rm th}$ ordered position, with $k \geq 3$ 
and $2 \le \ell \le k-1 $ known.
For every $\eps > 0$ and $\delta > 0$,
there is a constant $\xi > 0$ and an algorithm
that takes the stream over the interval $[0,T]$
with $p_T = \xi \log n \log \log n / n$ and
with probability at least $1 - \delta$ will
recover the correct value of $f(S)$ for at least a $1 - \eps$
fraction of all $k$-sets.
\label{passiveprop}
\end{theorem}

\begin{proof}
The algorithm operates in two phases similarly to the algorithm
for active queries, in that it first identifies the set $S^*$ of $k-1$
ineligible alternatives that cannot be selected in any
query.  Then it selects $k-2$ of these arbitrarily to form a set
$S^{-2}$ and analyzes sets of the form 
$\{u_i, u_j\} \cup S^{-2}$ to decide the relative ordering of
pairs $(u_i, u_j)$.
The key difference from the case of active queries is that 
the algorithm needs to operate using only the $k$-sets that
were handed to it initially at random, so it can't steer the queries
to seek out the set $S^*$, or run an adaptive sorting algorithm
on arbitrary pairs $(u_i, u_j)$.

Because the algorithm has two phases, we split the time
interval $[0,T]$ into two intervals $[0,T_1]$ and $(T_1, T_1 + T_2]$.
Note that for a given $k$-set $S$, the Union Bound implies that 
$p_{T_1} + p_{T_2}$ is an upper bound on the probability seeing $S$
during $[0,T]$, and hence 
$p_{T_1} + p_{T_2} \geq p_T$.

Given a time interval $[0,T]$ for generating the stream,
for a sufficiently large constant $b$ to be specified below,
we choose $T$ so that $p_T \geq b \log n / n + b \log n \log \log n / n$.
We split $[0,T]$ into $[0,T_1]$ and $(T_1, T_1 + T_2]$ so that
$p_{T_1} \geq b \log n / n$ and $p_{T_2} \geq b \log n \log \log n / n$.

Let $\D[a,b]$ denote the $k$-sets observed in the stream during 
the interval $[a,b]$.
We will use $\D_1 = \D[0,T_1]$ for the first phase, and
$\D_2 = \D (T_1,T_1+T_2]$ for the second phase.
Note that the contents of $\D_1$ and $\D_2$ are
independent of each other, an important facet 
of the streaming model.

\xhdr{Phase 1: Finding ineligible elements}
In the first phase, the algorithm will look at $\D_1$
and identify the set $S^{**}$
of elements $v \in U$ for which 
$v$ is not the output $f(S)$ for any $S \in \D_1$.
Clearly $S^* \subseteq S^{**}$, since no element of $S^*$ can be
the answer to any $f(S)$.
Now fix $v \not\in S^*$; what is the probability that it belongs to $S^{**}$?

Every eligible alternative $v$ will be the output $f(S)$ for some $S \in \D_1$
as long as there's a $k$-set in $\D_1$ containing exactly $\ell-1$ elements 
above $v$ and exactly $k-\ell$ element below $v$.
Note that since $2 \leq \ell \leq k-1$, this means
that there are a non-zero number of elements above and below $v$
in this case.
Let's call such a query a ``revealing query'' for $v$. 
Out of ${n \choose k}$ possible queries 
there are at least $n-2$ such revealing queries for $v$, 
since in the worst case there's only 
one way to choose $k-2$ alternatives above
and $n-2$ ways to choose one
alternatives below, or vice versa. 
Each revealing query for 
$v$ is chosen with probability $p \ge b \log n /n$.

The probability that none of these $n-2$ revealing queries 
is chosen is then
\begin{align*}
(1 - p)^{n-2} &\le(1 - b(\log n) /n)^{n-2} \le (1 - b(\log n) /n)^{n/2} 
\le e^{-(b/2) \log n} = n^{-b/2}.
\end{align*}
Now taking the Union Bound over all $n-k$ eligible alternatives 
$v$ says that the 
probability there exists a $v$ for which no revealing query is chosen
is bounded by $n \cdot n^{-b/2} \le n^{-(b/2)+1}$. 
Let $\ev_1$ denote the ``bad event'' that $S^* \neq S^{**}$, and
let $\delta_0 \leq \delta/4$ be a constant.
By choosing $b$ large enough,
we have $\Prb{\ev_1} \leq n^{-(b/2)+1} \leq \delta_0.$

\xhdr{Phase 2: Simulating pairwise comparisons}
In phase 2 we now use $k-2$ of the ineligible elements
as anchors for performing pairwise comparisons.
Recall that because we have split the stream over
two disjoint time intervals, the algorithm's behavior now is independent of its success or failure at finding $S^*$
in phase 1, so in our analysis we'll assume it has succeeded
in finding $S^*$.

For this phase, we use $p$ to denote $p_{T_2}$, where $T_2$ 
was chosen such that $p \geq b \log n \log \log n / n$.
To begin, we fix an arbitrary set $S^{-2}$ of $k-2$ elements from $S^*$, and
let $U' = U \setminus S^{-2}$. Let $G$ be the undirected graph on $U'$ in which $(u_i,u_j)$
is an edge if and only if $\{u_i, u_j\} \cup S^{-2}$ is a $k$-set in $\D_1$.
Note that $G$ is a uniform sample from the distribution
$G_{n-k+2,p}$, the Erd\H{o}s-R\'enyi random graph distribution.
By a theorem of Alon and Azar \cite{alon1988sorting}, it follows that if we compare
the pairs of elements $(u_i,u_j)$ defined by the edges of $G$,
then with probability at least $1 - \delta_1$
we will be able to infer the relative
order of at least a $1 - \gamma$ fraction of pairs of elements of $U'$,
where $\gamma$ and $\delta_1$ depend on $b$.
The remaining $\gamma$ fraction of pairs are simply left with their choice uninferred.
Let $\ev_2$ denote the ``bad event'' that the Alon-Azar algorithm does not succeed,
with $\Prb{\ev_2} \leq \delta_1$.
We choose $b$ large enough that
$\gamma \leq \eps / 2k^2$ and $\delta_1 \leq \delta / 4$.
Let $H$ be a directed acyclic graph on $U'$ in which 
$u_i$ points to $u_j$ if the Alon-Azar procedure has inferred
that $u_i$ is ranked ahead of $u_j$, and there is no edge between 
$u_i$ and $u_j$ otherwise.

Now we use the following procedure to answer queries of the
form $f(S)$.
If $S \not\subseteq U'$ then $S$ contains an ineligible element and we answer arbitrarily.
Otherwise, if $S \subseteq U'$, we look at all the pairwise
comparisons of elements in $S$ according to $H$.
If these pairwise comparisons form a complete acyclic digraph
then we choose the $\ell^{\rm th}$ element in order; otherwise we answer
arbitrarily.

For how many $k$-sets do we get the correct answer?
Let $P$ denote the set of node pairs in $U' \times U'$ for which
there is no edge in $H$.
If $\ev_2$ does not happen, then 
$|P| \leq \gamma {n - k + 2 \choose 2}.$
The number of $k$-sets that involve a pair from $P$ is thus at most
\small
\begin{align*}
\sum_{(u_i,u_j) \in P} {n - k \choose k - 2} 
\leq \gamma {n - k + 2 \choose 2} {n - k \choose k - 2} 
\leq \gamma k^2 {n \choose k}.
\end{align*}
\normalsize
Meanwhile, the number of $k$-sets that involve an element outside $U'$
is at most 
\small
\begin{align*}
(k-2) {n-1 \choose k-1} \leq (k-2) \left(\frac{k}{n}\right) {n \choose k}
\leq \gamma {n \choose k},
\end{align*}
\normalsize
provided 
that $n$ is large enough relative to $k$ that $k(k-2)/n \leq \gamma$.
But if a $k$-set $S$ satisfies $S \subseteq U'$, and $S \times S$
contains no pair of $P$, then all pairwise comparisons inside $S$
have been determined and thus we can report the correct
value of $f(S)$.  Hence if $\ev_2$ does not occur,
we report the correct value on all but at most
$\gamma (1 + k^2) {n \choose k} \leq \eps {n \choose k}$ $k$-sets, as desired.

\xhdr{The full algorithm and its analysis}
Now let's consider the full algorithm.
It starts by identifying a set $S^{**}$ of ineligible elements in phase 1.
If $S^{**}$ does not have $k-1$ elements, it terminates with no answer.
Otherwise, it runs the second phase on 
the assumption that $S^* = S^{**}$.

If the event $\ev_1 \cup \ev_2$ does not occur, then $S^* = S^{**}$ 
and the output of phase 2
satisfies the performance guarantee of Theorem \ref{passiveprop}.
With probability at least 
$1 - (\delta_1 + \delta_2) \geq 1 - \delta/2$, 
this union $\ev_1 \cup \ev_2$ does not occur,
in which case the full algorithm
satisfies the performance guarantee.
\end{proof}

The above proof holds for all values of $\ell$ except for
the two extremes of $\ell = 1$ and $\ell = k$, as stated.
This distinction is intriguing since the two extreme values 
are the ones that don't exhibit choice-set effects:
the individual making decisions is then selecting elements
according to a fixed total ordering of $U$.
The challenge with $\ell = 1$ and $\ell = k$ is actually that
the notion of ``ineligible'' elements becomes more subtle:
there are elements whose relative position in the embedding
can be resolved, but only with a very large number of passive queries.
Finding the right performance guarantee for this case is an
interesting open question.

\section{Distance comparison}
\label{s:distancecomparison}

All of our results thus far focus on recovering 
choice functions that are comparison-based
functions over ordered alternatives. 
As noted in the introduction,
this class of functions 
is surprisingly rich, with the ability 
to capture compromise effects in choices 
along any one-dimensional frontier, a choice-set effect
that cannot be exhibited by functions
that merely maximize over an ordering.
The preceding sections develop a theory of inference
 for such comparison-based functions.

In this section we shift our focus to recovering choice
functions defined by a different structural relationship, 
namely distance comparisons in a metric embedding. 
A choice function $f$ on $k$-sets is {\em distance-comparison-based} 
if the value $f(S)$ can be computed purely using comparisons
on the pairwise distances in a embedding $h$ with metric $d(\cdot,\cdot)$, 
i.e.~based
on comparisons on the set of pairwise distances 
$\{d(h(u_i),h(u_j)) : u_i, u_j \in S\}$. 

Distance comparisons form the basis of 
an additional important family of choice scenarios, 
namely choice queries answered with the ``most medial'' 
or the ``most distinctive'' alternative in a set, and in particular they
model an important choice-set effect 
called {\em similarity aversion} \cite{trueblood2013not}:
given a choice between two dissimilar alternatives $A$ and $B$
and one of two alternatives $A'$ and $B'$ that are similar to $A$ and $B$
respectively, similarity aversion arises when $f(\{A,B,B'\})=A$ but
$f(\{A,B,A'\})=B$. Essentially: given two similar alternatives and one dissimilar
alternative, the dissimilar option is chosen. 

It is important to note that similarity aversion cannot be modeled
by a choice function that is strictly comparison-based: such a function
can only evaluate ordinal comparisons, and when considering two
elements $A$ and $B$ embedded in one dimension, a comparison-based
function cannot resolve whether a third element $A'$ positioned
between $A$ and $B$ is closer to $A$ or closer to $B$.  

The remainder of the section has the following format. First
we show how similarity aversion can be modeled using 
distance-comparison-based choice functions. Next, 
we present two observations that suggest some of the difficulties that
must be overcome to learn such functions, a seemingly
more difficult inference context than comparison-based functions.
We then connect our observations to a broader literature on learning
metric embeddings, albeit without resolving our inference questions.

\vspace{\mynegvspace}
\xhdr{Similarity aversion from distance comparisons}
We now describe two choice functions of principle importance,
 and their capacity for being formulated as distance-comparison-based functions: 
 the median choice function and the outlier choice function. 
 We restrict our definitions to choices over $k$-sets
where $k$ is odd. The outlier choice function for triplets ($k=3$) will serve
as our model of similarity aversion.

The median choice function selects the element of a 
$k$-set S in $\mathbb R^m$ with metric $d(\cdot,\cdot)$
that minimizes the sum of distances to all other elements. 
In $m=1$ dimension this choice is the traditional median element,
and it can be selected according to the following distance-comparison
procedure:
repeatedly find the pair of alternatives 
$x, y \in S$ that are furthest apart and remove them from $S$. 
Since $k$ is odd, after
$(k-1)/2$ rounds there will be a single remaining element; return
that element as the choice. Thus we see that the median choice function
can be formulated as both a comparison-based
function---in the earlier sections, $q_{(k-1)/2,k}$---or 
as a distance-comparison-based function.

For $m>1$ dimensions and sets of size $k=3$ it is easy to verify
 that the above distance-comparison
procedure will still return the element that minimizes the sum of distances.
For $m>1$ and $k>3$, however, it is an open question whether
a distance-comparison-based function can return the element that minimizes the
sum of distances. Observe that this question
amounts to solving a version of 
the Fermat-Weber problem \cite{wesolowsky1993weber}
with a very restricted functional toolkit, and may be quite difficult.
We are not aware of any prior work on the capabilities of 
distance-comparison-based functions. Here we are 
principally interested in the special case of $k=3$, allowing us to skirt
this difficulty.

Building on the median choice function, 
the outlier choice function selects the element of a 
$k$-set $S$ that is farthest 
(in $\mathbb R^m$ using $d(\cdot,\cdot)$) from the median element 
(the median element being the element that
minimizes the sum of distances to all elements).
The outlier choice function can clearly be written as a 
distance-comparison-based function whenever 
the median choice function can. For $k=3$ elements 
it is precisely a model of similarity aversion on triplets.

As a general comment on the capabilities of distance-comparison-based functions,
it is clear that they can make choices that ordinary 
comparison-based functions cannot, but it is important to highlight
that the reverse is also true: the median choice function 
is the {\em only} position-selection function (in one dimension) 
that can be formulated as a distance-comparison-based function. A 
distance-comparison-based function cannot select a maximal or minimal 
element (or any off-center position), since it has no way of orienting
 a choice with regard to the direction of ``more'' (in any dimension).
The sets of comparison-based functions and distance-comparison-based
functions are therefore not neatly nested.

\begin{figure}[t]
\centering
\includegraphics[width =0.7\myfigwidth]{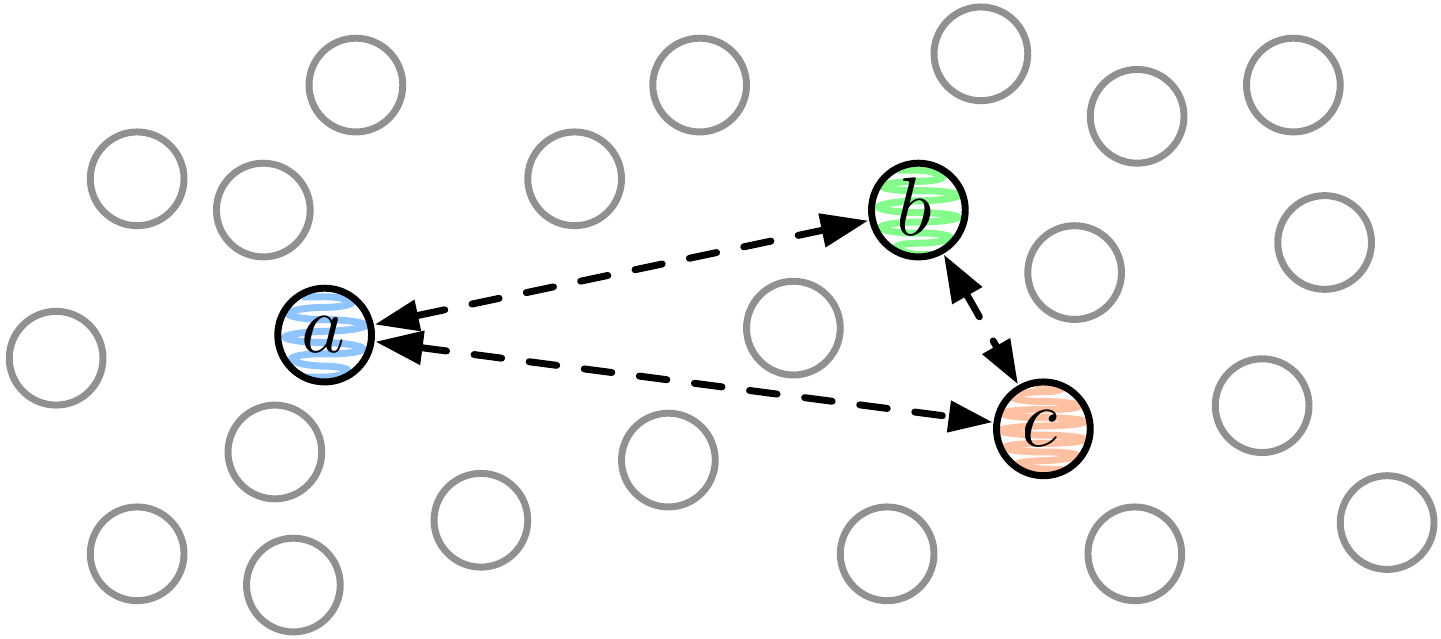}
\vspace{-5 pt}
\caption{
A $k$-set $S= \{a, b, c\}$ of $k=3$ alternatives embedded in $2$
dimensions. In this setting we can think of 
distance comparison as a comparison-based choice 
between three pairwise distances $\{d(a,b), d(b,c), d(a,c)\}$. 
When a distance-comparison-based choice function $f$ selects
the {\em outlier} $a$, this choice is effectively a comparison-based
choice of the shortest distance, $d(b, c)$.
When $f$ selects the {\em median element} $b$, this 
choice is effectively a comparison-based choice of the
longest distance, $d(a,c)$.
}
\label{f:fig3}
\vspace{-5pt}
\end{figure}

\vspace{\mynegvspace}
\xhdr{Distance comparison for triplets}
In the particular case of $k=3$, where choices are made over 
triplets $S=\{u,v,w\}$, we can think of a choice 
for an alternative $u \in S$  implicitly as a choice
for the distance $d(v,w)$ between the two remaining 
alternatives in $S$.
For a triplet $\{u,v,w\}$, the effective comparison-based query is over
the set of distances $\{d(u,v),$ $d(u,w),$ $d(v,w)\}$, where a choice of a
pairwise distance maps to a choice of the complementary
alternative. See Figure 3 for an illustration.
There are $N={n \choose 2}$ distances, and
so it may seem as though it would be possible to 
learn the ordering in $O(N \log N)$ queries in the active
query framework, with similar results carried over for other query
frameworks. Reality is more complicated.

We now give two observations on the difficulties of carrying
over comparison-based results to distance comparisons.
The first observation asks: when are there enough queries to learn 
the relative ordering of all the distances? The second observation
pertains to the restricted nature of the distance triplets that can be queried.

First, notice that each query for a $3$-set furnishes $2$ inequalities 
between distances.
There are ${n \choose 3}$ triplet queries for $n$ elements, which means that
querying for every possible $3$-set would produce $2{n \choose 3}$ inequalities.
Meanwhile there are ${n \choose 2}$ pairwise distances, 
meaning that there are ${n \choose 2}!$ possible permutations of the distances,
only one of which is the sorted order.
A sufficient condition for inferring
all choices is to know the sorted order of all the distances. 
Focusing on the outlier choice function (selecting the minimum distance that 
effectively returns the outlier element) in 
one does not need to know the relative order of the two largest distances. 
This tells us that it would take $\log_2({n \choose 2}!)-1$ bisections of the set of permutations
to identify the sorted order of all but the two largest distances.

We observe that $2{n \choose 3} < \log_2({n \choose 2}!)-1$ for $n \le 5$,
which tells us that for $n \le 5$ elements we will learn the choice for every
query before we can possibly know the sorted order of the relevant
pairwise distances. This tells us that a generic procedure that seeks to learn
all the pairwise distances will not always succeed. 

While our earlier query complexity results focus on asymptotic
complexity, the correctness of these algorithms holds for all $n$. 
The limiting observation presented here is specific to distance comparisons and 
does not apply to ordinary comparisons:
in that context there are only $n!$ permutations rather than ${n \choose 2}!$.

Second, observe that we cannot query for general triplets 
of distances $\{d(u,v),$ $d(w,x),$ $d(y,z)\}$, but are in fact
restricted to queries in the special case of $w=u$, $y=x$, and $z=v$. 
It is unclear for what $n$ it is possible to learn the relative ordering
of all the distances from such restricted distance comparisons.
While both of these initial observations are discouraging, 
an efficient inference algorithm for distance-comparison-based
functions would contribute a useful tool for learning in the presence
of similarity aversion and other choice-set effects, 
and we deem it important to present
distance-comparison-based functions as significant objects of study.

As an additional open direction for future work,
we briefly mention the choice-set effect known as the decoy effect
(or asymmetric dominance).
The decoy effect describes when the presence of a similar but inferior 
alternative increases the desirability of an option. The difference
between similarity aversion and the decoy effect is that the former
hinges on similarity while the latter
incorporates both similarity and inferiority. As such, modeling the decoy
effect requires a composition of both comparison and distance-comparison.
We leave the study of this alluring general class of choice functions,
which contains comparison-based and distance-comparison-based
choice functions as special cases, as future work.

\vspace{\mynegvspace}
\xhdr{Relationship to metric embeddings}
Beyond providing a model for 
similarity avoidance in the behavioral sciences, 
the ability to recover distance-comparison-based choice
functions also speaks to a broad literature on 
learning metric embeddings of data from distance-comparison-based
queries \cite{schultz2004learning}. 
In that literature, the {\em stochastic triplet embedding} technique
\cite{van2012stochastic} has recently been introduced as a
way to embed a generic dataset of elements
through answers to choices of the form ``Is A more similar to B or to C?''.
This question is effectively a comparison-based query requesting a choice
from the set $\{d(A,B),d(A,C)\}$. As an extension of that work,
the {\em crowd median algorithm} 
\cite{heikinheimo2013crowd} tries to
learn embeddings from triplet queries with the request
``Out of three shown items pick one that appears to be
different from the two others.''
The crowd median algorithm is employing precisely
the outlier choice function for $k=3$ described in this section.

Known work on the crowd median algorithm, however, leaves open the 
question of what embeddings it can learn. 
We observe that outside the behavioral modeling
that drove our work, a slight modification of the crowd median algorithm
to align with our results gives a method that can in fact learn embeddings
efficiently. We define the {\em generalized crowd median algorithm}, 
which instead asks questions
of the form ``of these $k$ pairs, which pair is least similar?''.
Under this more flexible query framework,
the ordering on the distance could then 
be inferred by the algorithms developed in this work:
for $n$ elements, it would then be possible in an active query 
framework to learn the ordering
of the $N={ n \choose 2 }$ distances in $O(N \log N)$ queries, for large $n$.

Our treatment of distance-comparison-based
queries also speaks to a broad line or research in
psychology on pair selection queries used to infer cognitive embeddings
\cite{smiley1979conceptual,ji2004culture,talhelm2015liberals}.
In these settings, sets of $k$ elements (where $k$ is often large) 
are presented to subjects with 
the question ``which two elements are most similar?''.
Such queries generalize the crowd-median
algorithm.
Lastly, very recent work has developed 
learning results for binary clustering (a simple form of embedding)
using so-called triangle queries of the 
form ``which of these birds belong to the same 
species?'' \cite{vinayak2016crowdsourced}, with possible 
relevance to learning from distance-comparison-based choices.

\section{Discussion}

The prevalence of choice-set effects in human decision-making
highlights a need for a principled inference model that can
support learning such effects.
In this work we've proposed a framework aimed at helping with 
such an integration, focussing on comparison-based choices.
A natural line of inquiry extending from choice-set effects quickly 
leads to rich issues related to the
theory of sorting, including
passive sorting
(``sorting in one round'')
and sorting with noisy information.

There are clearly many directions for further work.
We begin by mentioning a very concrete problem: while 
we studied passive queries and population mixtures, we did not
attempt to combine these two aspects to infer a mixed collection
of choice functions from a stream of passively observed queries.
A challenge in trying to achieve such a combination can be
seen in the technical ingredients that would need to be brought together;
the mechanics of the Feige et al.~procedure \cite{feige1994computing} 
and Alon-Azar procedure \cite{alon1988sorting} are not obviously compatible. 
Even the basic question of binary sorting with passive queries in the
presence of noisy information appears to be fairly open.

Many questions remain regarding how to capture relevant
population-level heterogeneities in large-scale
choice corpora.
Our framework for comparisons generally assumed a universally
agreed upon ordered embedding, and
it is very reasonable to suggest generalizations to
a dispersed distribution over possible embeddings.
A long line of work has brought distributions over rankings
into the literature on learning preferences
\cite{mallows1957non,fligner1986distance,braverman2009sorting},
including through comparisons \cite{lu2011learning}, 
and it would be natural to explore analogous 
distributional modeling for the embedding
that underlies comparison-based choice functions.
Analyzing mixtures of comparison-based choice functions 
as a model of probabilistic discrete choice \cite{baltas-doyle}
is another closely related open research direction.

One can also consider generalizing the embedding that defines
the structure of the alternatives and the choice sets.
For example, we could think about the alternatives as being embedded not
just in one dimension (or a one-dimensional Pareto frontier)
but in multiple dimensions,
and an individual could execute a sequence of comparison-based 
rules to select an item from a choice set.
Such a generalization could provide a way to incorporate
a number of other well-documented choice-set effects
\cite{trueblood2013not}, and could
form intriguing potential connections
to the {\em elimination by aspects} \cite{tversky-eba} model
of discrete choice. 

We find it encouraging that in the face of
choice behavior much more complex
than ordered preferences,
we are able to fully match most of the known
query complexity results from the theory of sorting
from binary comparisons. 
Continued work to develop a more complete
theory of inference
for 
other 
increasingly rich classes of choice functions has
the potential to lay a foundation for a much needed unified
theory for learning choice-set effects.

\bibliographystyle{plain}
\bibliography{choicerefs}

\end{document}